\documentclass[conference]{IEEEtran}
\IEEEoverridecommandlockouts
\usepackage{cite}
\usepackage{amsmath,amssymb,amsfonts,mathtools}
\usepackage{algorithmic}
\usepackage{graphicx}
\usepackage{textcomp}
\usepackage{xcolor}
\usepackage{subcaption}
\def\BibTeX{{\rm B\kern-.05em{\sc i\kern-.025em b}\kern-.08em
    T\kern-.1667em\lower.7ex\hbox{E}\kern-.125emX}}

\def\psfancypar#1#2{\begingroup\def\par{\endgraf\endgroup\lineskiplimit=0pt}
               \setbox2=\hbox{\large\sc #2}
               \newdimen\tmpht \tmpht \ht2 \advance\tmpht by \baselineskip
               \font\hhuge=cmti12 at \tmpht
              \setbox1=\hbox{{\hhuge #1}}
               \count7=\tmpht \count8=\ht1
               \divide\count8 by 1000 \divide\count7 by \count8
               \tmpht=.001\tmpht\multiply\tmpht by \count7
               \font\hhuge=cmbx10 at \tmpht
               \setbox1=\hbox{{\hhuge #1}}
               \noindent
                \hangindent1.05\wd1
               \hangafter=-2 {\hskip-\hangindent
               \lower1\ht1\hbox{\raise1.0\ht2\copy1}%
                \kern-0\wd1}\copy2\lineskiplimit=-1000pt}

\newcommand{\boxit}[1]{\vbox{\hrule\hbox{\vrule\kern6pt
\vbox{\kern6pt#1 \vspace{-10 pt}\kern6pt}\kern6pt\vrule}\hrule}}
\newcommand{\boxitsl}[1]{\vbox{\hrule\hbox{\vrule\kern6pt
\vbox{\kern6pt#1 \vspace{3 pt}\kern6pt}\kern6pt\vrule}\hrule}}
\def\thetabf{{\mbox{\boldmath$\theta$\unboldmath}}}

\newcommand{\E}{\mbox{{\rm E}}}

\def\reals{ { {\rm  I \kern-0.15em R }  } }
\def\complex{\hbox{ \,{{\rm C} \kern-0.50em \raise0.20ex {  |}}\,}}

\def\Pibf{\hbox{$\boldsymbol{\Pi}$}}

\def\abf{\hbox{\bf a}}

\def\cbf{\hbox{\bf c}}

\def\ebf{\hbox{\bf e}}

\def\nbf{\hbox{\bf n}}

\def\sbf{\hbox{\bf s}}

\def\vbf{\hbox{\bf v}}

\def\xbf{\hbox{\bf x}}
\def\ybf{\hbox{\bf y}}

\def\xbf{\hbox{\bf x}}
\def\ybf{\hbox{\bf y}}

\def\Bbf{\hbox{\bf B}}

\def\Dbf{\hbox{\bf D}}

\def\Hbf{\hbox{\bf H}}
\def\Ibf{\hbox{\bf I}}
\def\Jbf{\hbox{\bf J}}

\def\Pbf{\hbox{\bf P}}
\def\Qbf{\hbox{\bf Q}}

\def\Sbf{\hbox{\bf S}}

\def\Xbf{\hbox{\bf X}}

\def\Zbf{\hbox{\bf Z}}

\def\be{\vskip .3cm \begin{equation}}
\def\ee{\end{equation} \vskip .4cm \noindent}

\newcounter{remarknr}
\renewcommand{\theremarknr}{\arabic{remarknr}}

\newtheorem{theorem}{Theorem}

\newtheorem{corollary}{Corollary}

\newcounter{assumpnr}
\renewcommand{\theassumpnr}{\arabic{assumpnr}}

\newcommand{\pbox}{\hfill $\Box$}

\newcounter{examplenr} 
\renewcommand{\theexamplenr}{\arabic{examplenr}}

\begin{document}

\title{\large 
STATISTICAL ANALYSIS OF TARGET PARAMETER ESTIMATION USING PASSIVE RADAR
\thanks{The research was funded by the strategic innovation programme “Smartare Elektroniksystem”, a joint research project financed by VINNOVA, Formas and the Swedish Energy Agency, and by SAAB.}
}

\author{\IEEEauthorblockN{Mats Viberg\thanks{Mats Viberg is also with Chalmers University of Technology.}}
\IEEEauthorblockA{
\textit{Blekinge Institute of Technology}\\
Karlskrona, Sweden 
}
\and
\IEEEauthorblockN{Daniele Gerosa}
\IEEEauthorblockA{
\textit{Chalmers University of Technology}\\
Göteborg, Sweden 
}
\and
\IEEEauthorblockN{Tomas McKelvey}
\IEEEauthorblockA{
\textit{Chalmers University of Technology}\\
Göteborg, Sweden 
}
\and
\IEEEauthorblockN{Thomas Eriksson}
\IEEEauthorblockA{
\textit{Chalmers University of Technology}\\
Göteborg, Sweden 
}
}

\maketitle

\begin{abstract}
A passive radar system uses one or more so-called Illuminators of Opportunity (IO) to detect and localize targets. In such systems, a reference channel is often used at each receiving node to capture the transmitted IO signal, while targets are detected using the main surveillance channel. The purpose of the present contribution is to analyze a method for estimating the target 
parameters in such a system. Specifically, we quantify the additional error contribution due to not knowing the transmitted IO waveform perfectly. A sufficient condition for this error to be negligible as compared to errors due to clutter and noise in the surveillance channel is then given.
\end{abstract}



\section{Introduction}


A passive radar system takes advantage of existing electromagnetic signals to detect and locate targets. The so-called Illuminators of Opportunity (IO) can be TV or radio transmitters, or even satellites
\cite{kuschel_tutorial_2019}. Such systems have received much attention in the signal processing and radar system community, due to advantages in terms of power consumption and covertness among others \cite{kuschel_tutorial_2019,greco_passive_2023}.

In passive radar, the lack of knowledge of the transmitted signal poses a challenge, see e.g. \cite{tong_cramerrao_2019,zhang_sparsity-based_2019,zhang_maximum_2020} for approaches to jointly estimate this signal and the target parameters. We assume here that the Receiver Nodes (RN) are equipped with a separate Reference Channel (RC), directed towards the IO transmitter. In essence, the (noisy) RC signal then replaces the IO waveform when estimating the target parameters from the Surveillance Channel (SC) data. Such a setting is considered, e.g. in \cite{gogineni_passive_2018}, where also a technique to handle multipath (clutter) in the RC is proposed, using the principal component of the cross-correlation matrix between the RC and the RC signals. In \cite{zhang_maximum_2019}, the Maximum Likelihood (ML) estimator is derived for this scenario, assuming both channels to be corrupted by White Gaussian Noise (WGN). In the present contribution, we consider the case of WGN in the RC, whereas the SC is corrupted by Direct-Path Interference (DPI) as well as Clutter Interference (CI).
We consider the problem of estimating the target position and velocity in a 2D scenario, and the goal is to assess the estimation performance in the presence of noise in both receiving channels. The estimator under study is based on the Extensive Cancellation Approach (ECA) of \cite{colone_multistage_2009}, which estimates target parameters using a matched-filtering approach after canceling the DPI and CI. 
A similar approach is also pursued in \cite{zhou_direct_2024}, where it is compared to the full ML that uses both channels to estimate the IO signal. The comparison in \cite{zhou_direct_2024} is based on the Cramér-Rao Lower Bound (CRLB) for the localization problem. Our approach is based on a first-order statistical analysis of the estimated target parameters, assuming ``high enough" Signal-to-Noise Ratio (SNR) in both the RC and the SC channels. We derive an explicit expression of the covariance matrix of the estimation error of the target delay and Doppler parameters. This is recognized as the sum of the CRLB assuming a perfectly known transmitted signal, and an additional term due to the noise in the RC. This enables us to give sufficient conditions on the required SNR in the RC, in
order for the IO signal errors to be negligible as compared to the errors due to interference and clutter in the SC.



\section{Problem Description}

We consider a passive radar scenario where $K$ RN:s collect data emanating from reflections of an unknown IO signal. 
For simplicity, we assume a single IO and a single target.
The extension to multiple IO:s is straightforward provided they transmit orthogonal waveforms as in, e.g., \cite{tajer_optimal_2010}. We also remark that the presented method can be applied to multi-target scenarios by searching for multiple peaks in the global likelihood function. Iterative approaches similar to \cite{yi_suboptimal_2020} are also possible.

We assume the position of the IO to be known at each RN. 
The target of interest is moving at constant non-zero speed, and the goal is to estimate its position and velocity. The DPI and CI are assumed stationary during the observation interval. The RN:s transmit data to a central node (CN). The data from different RN:s are synchronized to time-delay (relative bandwidth), but not to phase (relative carrier), see e.g. \cite{m_weiss_synchronisation_2004}.

\section{Data Model}

The IO transmits a signal $\Re\{s(t)e^{j\omega_{c}t}\}$, where $s(t)$ is the baseband complex envelope and $\omega_c$ the carrier frequency. 
The available baseband data after demodulation at receiver node $k,\ k=1,\dots,K$, is modeled by
\begin{align}
    x_{k}(t) &= a_{k} s_k(t) + n_{k}(t) \label{eq:ref} \\
    y_{k}(t) &= b_{k} s_k(t) +  
    y_k^c(t)
    + d_{k} s_k(t-\tau_{k}) e^{j\omega_k t} + e_{k}(t)\,, \label{eq:surv}
\end{align}
where $x_k(t)$ and $y_k(t)$ represent the data from the RC and the SC respectively. 

In both equations above, $s_k(t)$ is the IO waveform arriving at node $k$, thus serving as a reference for that node. 
The RC signal
(\ref{eq:ref}) consists of the direct path, where $a_{k}$ accounts for propagation attenuation and receiver characteristics, and receiver noise $n_{k}(t)$. In the SC (\ref{eq:surv}), the first term is the DPI, the second contains clutter from stationary objects,
the third is the target reflection, and the last term is receiver noise. The complex target amplitude $d_k$ accounts for propagation attenuation, bistatic Radar Cross Section (RCS), as well as antenna and receiver gain. The factor $e^{j\omega_k t}$ represents the Doppler effect due to the target motion, which is assumed to be with constant velocity during the data collection interval.
The target model in (\ref{eq:surv}) assumes a low speed, a narrowband baseband signal, and a sufficiently short observation window.
Note that the time-delay and Doppler parameters are related to the target position and velocity through the known locations of the various nodes. 

Let the data collection time in the SC be $0\leq t < T$, in which $N$ samples are collected at time instances $t_n = n\Delta T$, $\Delta T = T/N$, $n=0,\dots,N-1$.
Due to the time-delay of the clutter and target components, it is assumed that the data collection starts earlier for the RC. Let the maximum delay of interest be $M$ samples. 
The available data samples are then
$\{x_{k}(t_n)\}_{n=-M}^{N-1}$ and
$\{y_{k}(t_n)\}_{n=0}^{N-1}$, which
can be put into vector form as
\begin{align}
    \xbf_{k}^R &= a_{k}^{} \sbf_k^{R} + \nbf_{k}^R \label{eq:vref} \\
    \ybf_{k} &= b_{k} \sbf_k + 
    \ybf_k^c +
    d_{k}\, \sbf_k(\tau_{k}) \odot \vbf(\omega_k) + \ebf_{k}\,, \label{eq:vsurv}
\end{align}
where 
$\xbf_{k}^R,\ \sbf_{k}^{R}$ and $\nbf_{k}^{R}$ are column vectors in $ \mathbb{C}^{M+N}$, and
$\ybf_k,\ \sbf_k,\ \ybf_k^c,\ 
\sbf_k(\tau_k)$ and $\ebf_k$ are in $\mathbb{C}^{N}$.
We use $\sbf_k(\tau_k)$ to denote the IO waveform $\sbf_k$ with a time-delay $\tau_k$.
Further, the target Doppler is modeled by the DFT vector $\vbf(\omega_k)$,
\begin{equation}
    \vbf(\omega_k) = 
    [1,e^{j\omega_k \Delta T},  \dots,e^{j\omega_k (N-1) \Delta T}]^T,
\end{equation}
and $\odot$ represents the Schur/Hadamard product (elementwise multiplication). 
In general, the clutter component is an integral over reflections from all illuminated distances. However, we assume here that $\ybf_k^c$ can be approximated as a linear combination of a finite set of past samples of the IO signal:
$$
\ybf_k^c = \sum\limits_{l\in \mathcal{L}_k}
c_l \sbf_k(l\Delta T)
=\sum\limits_{l=1}^L
c_l \sbf_k(l\Delta T)\,,
$$
where the set $\mathcal{L}_k$, of cardinality $L$, $L\leq M$, contains the time delays $l$ for which we expect clutter returns. For the sake of simplicity, we assume in the second equality that $\mathcal{L}_k=\{1,\dots,L\}, \forall k$.
Hence, we can express the clutter vector in (\ref{eq:vsurv}) as
\begin{equation}
    \ybf_k^c = \Sbf_k \cbf_k\,,
\end{equation}
where $\Sbf_k$ is an $N\times L$ Toeplitz matrix containing samples of $s_k(t_n)$ for 
$-L \leq n \leq N-2$,
and where $\cbf_k = [c_1,\dots,c_L]^T$ is the vector of FIR filter coefficients. It is assumed that $L<N-1$ and that $\Sbf_k$ has full column rank.
The noise samples 
$\{n_k(t_n)\}_{n=-L}^{N-1}$ and $\{ e_k(t_n)\}_{n=0}^{N-1}$ are all 
assumed to be i.i.d. WGN with variances $\sigma_n^2$ and $\sigma_e^2$ respectively. 

\section{Target Parameter Estimation}

The estimation method to be investigated has the following steps, similar to \cite{colone_multistage_2009}. First, each RN uses the RC output as the ``true" IO waveform, next this is used to cancel the DPI and CI in the SC, and finally the target time-delay and Doppler parameters are estimated. These are fed to the CN, which performs the matching to the target position and speed parameters. 

\subsection{Interference Cancellation}

Under the above assumptions and modeling the IO waveform vector $\sbf_k$ as deterministic and unknown \cite{zhang_maximum_2019,zhou_direct_2024}, the global Maximum Likelihood (ML) approach leads to a combination of likelihood functions from each node. At node $k$, the negative log-likelihood function, ignoring constants, is given by
\begin{equation}
    \frac{1}{\sigma_n^2} \left\|
\xbf_{k}^R - a_{k} \sbf_k^{R} \right\|^2 + \frac{1}{\sigma_e^2} \left\|
    \ybf_{k} - b_{k} \sbf_k - \Sbf_k \cbf_k -
    d_{k}\, \abf(\tau_k,\omega_k)
    \right\|^2 ,
    \label{eq:ML1}
\end{equation}
where we have introduced the ``steering vector"
\begin{equation}
\abf(\tau_k,\omega_k) =
\sbf_k(\tau_{k}) \odot \vbf(\omega_k)\, .
\end{equation}

The ML estimate is now found by minimizing (\ref{eq:ML1}) with respect to all unknown parameters.
Since the SNR of the IO signal is assumed to be much stronger in the RC than in the SC, the IO waveform can be determined with a good approximation by only using the RC data \cite{colone_multistage_2009,  zhou_direct_2024}. The estimated IO signal part is then applied to the second term of (\ref{eq:ML1}) as $\hat{\sbf}_k=\xbf_k$,
$\hat{\Sbf}_k = \Xbf_k $, 
and $\hat{\sbf}_k(\tau_k)=\xbf_k(\tau_k)$,
respectively, where the sampling instances in $\xbf_k$ and $\Xbf_k$ have been synchronized with those in $\sbf_k$ and $\Sbf_k$, and where the unknown amplitude $a_k$ has been absorbed into the IO waveform estimate.
Inserted into (\ref{eq:ML1}), this yields
\begin{equation}
\ell(b_k,\cbf_k,d_k,\tau_k,\omega_k) = \left\|
    \ybf_{k} - b_{k} \xbf_k - \Xbf_k \cbf_k -
    d_{k}\, \hat{\abf}(\tau_k,\omega_k)
\right\|^2 ,
    \label{eq:ML2}
\end{equation}
where the estimated steering vector is defined as
\begin{equation}
\hat{\abf}(\tau_k,\omega_k)
=    \xbf_k(\tau_{k}) \odot \vbf(\omega_k)\, .
\label{eq:ahatdef}
\end{equation}

Minimizing (\ref{eq:ML2}) w.r.t. $b_k$ and $\cbf_k$ and substituting the resulting estimates back into (\ref{eq:ML2}) results in an effective cancellation of the direct IO and the clutter interference. To this end we introduce the noise-free and noisy interference matrices as
$$
\Sbf_I = [\sbf_k,\Sbf_k]\ ,\quad \Xbf_I = [\xbf_k,\Xbf_k]\, .
$$
The orthogonal projection matrices onto the orthogonal complements of the span of $\Sbf_I$ and $\Xbf_I$ are given by
\begin{align}
    {\Pibf}^{\perp} &=  \Ibf - {\Pibf}
    =  \Ibf -
    \Sbf_I\left(\Sbf_I^{H}\Sbf_I\right)^{-1}\Sbf_I^H
    \label{eq:Piperp} \\
    \hat{\Pibf}^{\perp} &=  \Ibf - \hat{\Pibf}
    =  \Ibf -
    \Xbf_I\left(\Xbf_I^{H}\Xbf_I\right)^{-1}\Xbf_I^H   .
    \label{eq:Piperphat}
\end{align}

With these definitions, the minimum of (\ref{eq:ML2}) w.r.t. $b_k$ and $\cbf_k$ reduces to the interference-cleaned version
\begin{equation}
    \ell(d_k,\tau_k,\omega_k) =
     \left\| \hat{\Pibf}^{\perp} \left( \ybf_k -
    d_{k}\, \hat{\abf}(\tau_k,\omega_k) \right) \right\|^2 .
\label{eq:ML3}
\end{equation}
Substituting the minimizing $d_k$ from (\ref{eq:ML3}) back into the criterion then results in the final form
\begin{equation}
    \ell(\tau_k,\omega_k) = 
    \left\| \hat{\Pibf}^{\perp}\ybf_k  -
\hat{\Pibf}^{\perp}\frac{\hat{\abf}(\tau_k,\omega_k) \hat{\abf}^H(\tau_k,\omega_k)}{\hat{\abf}^H(\tau_k,\omega_k)\hat{\Pibf}^{\perp}\hat{\abf}(\tau_k,\omega_k)}\, \hat{\Pibf}^{\perp}\ybf_k\right\|\, .
\end{equation}
Clearly, minimizing $\ell(\tau_k,\omega_k)$ is equivalent to maximizing the following interference-canceled and normalized version of the 2D delay-Doppler ambiguity function:
\begin{equation}
P_k(\tau_k,\omega_k) =
\frac{ |\hat{\abf}^H(\tau_k,\omega_k) \hat{\Pibf}^{\perp} \ybf_k|^2}{\hat{\abf}^H(\tau_k,\omega_k)\hat{\Pibf}^{\perp} \hat{\abf}(\tau_k,\omega_k)} =
\ybf_k^H \hat{\Pbf}
\, \ybf_k^{}
\, ,
\label{eq:ML4}
\end{equation}
where 
$\hat{\Pbf}$
is the orthogonal projection matrix onto the range space of $\hat{\Pibf}^{\perp}\hat{\bf a}(\tau_k,\omega_k)\,$:
\begin{equation}
\hat{\Pbf}= 
\frac{\hat{\Pibf}^{\perp}\hat{\abf}(\tau_k,\omega_k) \hat{\abf}^H(\tau_k,\omega_k)\hat{\Pibf}^{\perp}}{\hat{\abf}^H(\tau_k,\omega_k)\hat{\Pibf}^{\perp}\hat{\abf}(\tau_k,\omega_k)} = \Ibf - 
\hat{\Pbf}^{\perp}\, .
\label{eq:Pdef}
\end{equation}
Although (\ref{eq:ML4}) depends implicitly on the target position and velocity, the $k$:th receiver node can only evaluate the criterion with respect to the target delay and Doppler parameters relative its own position. Thus, at node $k$, (\ref{eq:ML4}) is computed on a $(\tau_k,\omega_k)$-grid, and the ``significant" values are transmitted to the central node for further processing. Though not considered in this paper, it is noted that the final detection decision should be made only after combining all RN information at the CN.

\subsection{Target Localization}

The final step is to combine the delay-Doppler information from all receiver nodes using a global Maximum Likelihood approach at the central node. 
The available data are the sampled versions of (\ref{eq:ML4}) from all nodes. Since the data are independent, the global likelihood function simply adds all contributions: 
\begin{equation}
V_{ML}(\thetabf) = \sum_{k=1}^K  P_k(\tau_k(\thetabf),\omega_k(\thetabf))\, ,
\label{eq:GlobalML}
\end{equation}
where $\tau_k=\tau_k(\thetabf)$ and $\omega_k=\omega_k(\thetabf)$ are known functions of the 4D target parameter vector $\thetabf$, which contains the $(x,y)$ coordinates as well as its speed in the $x$ and $y$ directions. In case the noise variances are different among the RN:s, their respective contribution should be suitably weighted in (\ref{eq:GlobalML}).
The global ML estimator is now to perform a search of (\ref{eq:GlobalML}) over the target parameters in 4 dimensions. 
The search can be mitigated by 
using only values of (\ref{eq:ML4}) that exceed a certain threshold. 
For each hypothesized target localization and velocity, the corresponding time-delay and Doppler parameters are calculated for each node. The resulting value of (\ref{eq:ML4}) is added to the global likelihood function, and if this sample is missing at a particular node, we simply add zero. It should be noted that the discretizations of $(\tau_k,\omega_k)$ at the different RN:s are not synchronized, and care must be taken when combining their information \cite{yang_multitarget_2022}.

\section{Statistical Performance Analysis}

The proposed global ML estimator is approximate in the sense that it uses the reference channel as if it were the true IO waveform. It is of interest to quantify analytically the effect of this approximation. Specifically, what SNR is required in the reference channel in order for the approximate ML estimates to achieve the Cramér-Rao lower bound for the target parameters, assuming a perfect IO waveform knowledge? 

To answer this question, we establish the first-order covariance matrix of the estimated target parameters, assuming a ``high enough" SNR in both the RC and the SC at each node. The first step is to establish consistency, in the sense that for noise-free data, the criterion function (\ref{eq:ML4}) is maximized by the true target delay-Doppler pair, as seen from the $k$:th RN. It is easy to establish this result if the IO signal is such that the steering vector is \textit{unambiguous}. By this we mean that the following holds true over the range of target parameters of interest:
\begin{equation}
\sbf(\tau_{0}) \odot \vbf(\omega_0) = \sbf(\tau) \odot \vbf(\omega)\  \Longleftrightarrow \ (\tau_0,\omega_0)=(\tau,\omega)\, .
\label{eq:unambiguous}
\end{equation}
In order to express the approximate covariance matrix in a compact form, we introduce the following notation. First, let
\begin{equation}
\Dbf_k = \left[ \frac{\partial \abf(\tau_k(\thetabf),\omega_k(\thetabf))} {\partial\theta_1},\dots,\frac{\partial \abf(\tau_k(\thetabf),\omega_k(\thetabf))} {\partial\theta_4}
\right]
\end{equation}
denote the matrix of derivatives of the $k$:th steering vector with respect to the target parameters. Further, define the $N\times N(L+1)$ matrix
\begin{equation}
\Zbf_k = \left[
b_k\,\Ibf + d_k\,\text{diag}(\vbf(\omega_k))\quad \cbf_k^T\otimes \Ibf
\right]\, ,
\label{eq:Zdef}
\end{equation}
and let $\Jbf_k$ be a selection matrix such that
$$
\text{vec}(\Sbf_I) = \Jbf_k \sbf_I\,,
$$
where $\sbf_I=[
s_k(t_{-L}),\dots,s_k(t_{N-1})
]^T$. 
The three components in (\ref{eq:Zdef}) model the error contributions to $\hat{\thetabf}$ due to not knowing the IO waveform perfectly. The first term is due to imperfect DPI cancellation, the second handles the effect of using the incorrect steering vector in (\ref{eq:ML4}), i.e. a ``mismatched filter", while the last term comes from the imperfect clutter cancellation. 

Further, let $\Pbf$ be the noise-free version of $\hat{\Pbf}$ and
introduce the orthogonal projection matrix
\begin{equation}
\Tilde{\Pbf} = \Pibf^{\perp} \Pbf^{\perp}\Pibf^{\perp},
\end{equation}
which projects onto $\text{span}(\Pibf^\perp\Bbf)$, where $\Bbf$ is any matrix that spans the orthogonal complement of the steering vector $\abf$. We can now state the main result of this paper.
\begin{theorem}
Let $\hat{\thetabf}$ be obtained by maximizing (\ref{eq:GlobalML}) and assume the IO waveform be such that (\ref{eq:unambiguous}) holds. Assume further that the radar scenario is such that the 
$\{(\tau_k,\omega_k)\}_{k=1}^K$ pairs together uniquely determine $\thetabf$. Then, as $\sigma_n^2\rightarrow 0$ and $\sigma_e^2\rightarrow 0$ jointly, we have $\hat{\thetabf} \rightarrow\thetabf_0$ in probability; and its covariance matrix is to first order given by 
\begin{equation}
\E [(\hat{\thetabf} - \thetabf_0) (\hat{\thetabf} - \thetabf_0)^T ] \approx \textbf{CRB}_{\boldsymbol{\theta}} +
\Hbf^{-1} \Qbf \,\Hbf^{-1} ,
\label{eq:CRBascov} 
\end{equation}
where
\begin{align}
\textbf{CRB}_{\boldsymbol{\theta}}  &= \sigma_e^2\,\Hbf^{-1}
= \frac{\sigma_e^2}{2} \left( \sum\limits_{k=1}^K |d_k|^2\, \Re \left\{ \Dbf^H _k\Tilde{\Pbf}\,\Dbf_k^{}\right\} \right)^{-1} \\
\Qbf &= 2\sigma_n^2\,\,\sum\limits_{k=1}^K \frac{|d_k|^2}{|a_k|^2}\, \Re \left\{ \Dbf^H_k \Tilde{\Pbf}\,\Zbf_k^{}\Jbf_k^{} \Jbf_k^T\Zbf_k^{H}\Tilde{\Pbf}\,
\Dbf_k^{}\right\} \,.
\label{eq:Q}
\end{align}
The term $\textbf{CRB}_{\boldsymbol{\theta}}$ in (\ref{eq:CRBascov}) is the CRLB for $\thetabf$ assuming a perfectly known IO signal, and the second term quantifies the excess error due to the noise in the reference channel.
\end{theorem}

\begin{proof}
The global ML estimate of the target parameters are obtained by maximizing (\ref{eq:GlobalML}) w.r.t. $\thetabf$. Thus, treating $\thetabf$ as a continuous-valued parameter, the gradient is zero at the optimal value,
\begin{equation}
V'_{ML}(\hat{\thetabf}) = 0\, ,
\label{eq:gradient}
\end{equation}
where $\hat{\thetabf}$ denotes the ML estimate.
A first-order Taylor expansion of (\ref{eq:gradient}), (see, e.g., Theorem 2 in \cite{viberg_ottersten_1991}), then eventually leads to (\ref{eq:CRBascov}) -- (\ref{eq:Q}).
Due to space limitations, the details are deferred to a future publication. \pbox
\end{proof}


Using Theorem 1, we can now address the question of how large the errors in the reference signal can be in order for its effect to be negligible. 
Upon comparing the two corresponding terms in (\ref{eq:CRBascov}), we can conclude that the matrix $\Zbf_k^{}\Jbf_k^{} \Jbf_k^T\Zbf_k^{H}$ plays a crucial role. Specifically, 
the second term is negligible if it holds that
$$
\left\| \Zbf_k\Jbf_k \right\|_2^2 \ll \sigma_e^2\,|a_k|^2 / \sigma_n^2 \quad \forall k\,.
$$
This leads to the following corollary to Theorem 1.
\begin{corollary}
Let the assumptions in Theorem 1 hold. Then, the second term in (\ref{eq:CRBascov}) can be neglected if the following holds true for all $k$:
\begin{equation}
(L+1)\,\frac{|b_k|^2 + |d_k|^2 + \|\cbf_k\|^2}{\sigma_e^2} \ll
\frac{|a_k|^2}{\sigma_n^2}\, .
\label{eq:interferencebound}
\end{equation}
\end{corollary}
The right-hand side of (\ref{eq:interferencebound}) is the SNR in the RC, and the left-hand side is an upper bound on the total interference-to-noise ratio in the SC. We recall that the second term is due to noise in the steering vector (\ref{eq:ahatdef}) when used in (\ref{eq:ML4}). Thus, its effect is proportional to the target power, which may seem counterintuitive.
We note that this bound guarantees that the parameter estimates achieve the CRLB for a known transmitted signal for \textit{all} values of the target parameters, including targets moving at a very slow speed. For faster targets, there is an additional suppression of the clutter due to the matched-filtering, which leads to a looser requirement than (\ref{eq:interferencebound}).


\section{Numerical Examples}
In this section, we compare the theoretical error approximation in \eqref{eq:CRBascov} with results from a Monte Carlo (MC) simulation. To keep the presentation simple, we set \(K=1\), thus using only \( \thetabf = (\tau, \omega) \) as target parameters. The IO test waveform is complex-valued and bandlimited, with bandwidth \( = 8 \text{ MHz} \), unitary power and carrier frequency of \( 600 \text{ MHz} \). The sampling frequency at baseband is \(25 \text{ MS/s}\). The IO signal propagates deterministically in space, and the amplitudes \( a_k \), \( b_k \) and \( d_k \) are determined by the standard bistatic radar equation applied to the geometry under investigation. The clutter filter length is chosen as \(L=70\), and the (single) target is moving at \( \approx 250 \text{ m/s} \). Its starting position is chosen at random (and then kept fixed), so that it falls \emph{within the clutter range}. The acquisition time is chosen to be \( N = 2^{13} \) samples, corresponding to approximately 0.33 ms of data. We present two plots using 
different SC and RC SNR levels respectively, and we plot only one parameter per SNR ``type''. The cost function in \eqref{eq:ML4} is maximized for \(2500\) noise realizations using the Nelder-Mead algorithm implemented in MATLAB's \texttt{fminsearch}. The other quantities at play are kept constant; in particular we had \(75 \) dB RC SNR in Figure \ref{fig:RMSEtau} and \(15\) dB SC SNR in Figure \ref{fig:RMSEomega}. The SNR here is calculated either as \( |d_1|^2 /\sigma_e ^2  \) or as \(|a_1|^2 /\sigma_n ^2 \) in dB. Figure \ref{fig:RMSEtau} shows a tight agreement between MC simulations and the theory down to \mbox{-20} dB SC SNR, and Figure \ref{fig:RMSEomega} highlights the importance of the correction term in \eqref{eq:CRBascov} when the RC SNR is not sufficiently high.
\begin{figure}[htbp]

  \centering
  \begin{subfigure}[b]{0.47\textwidth}
    \centering
    \includegraphics[width=\linewidth]{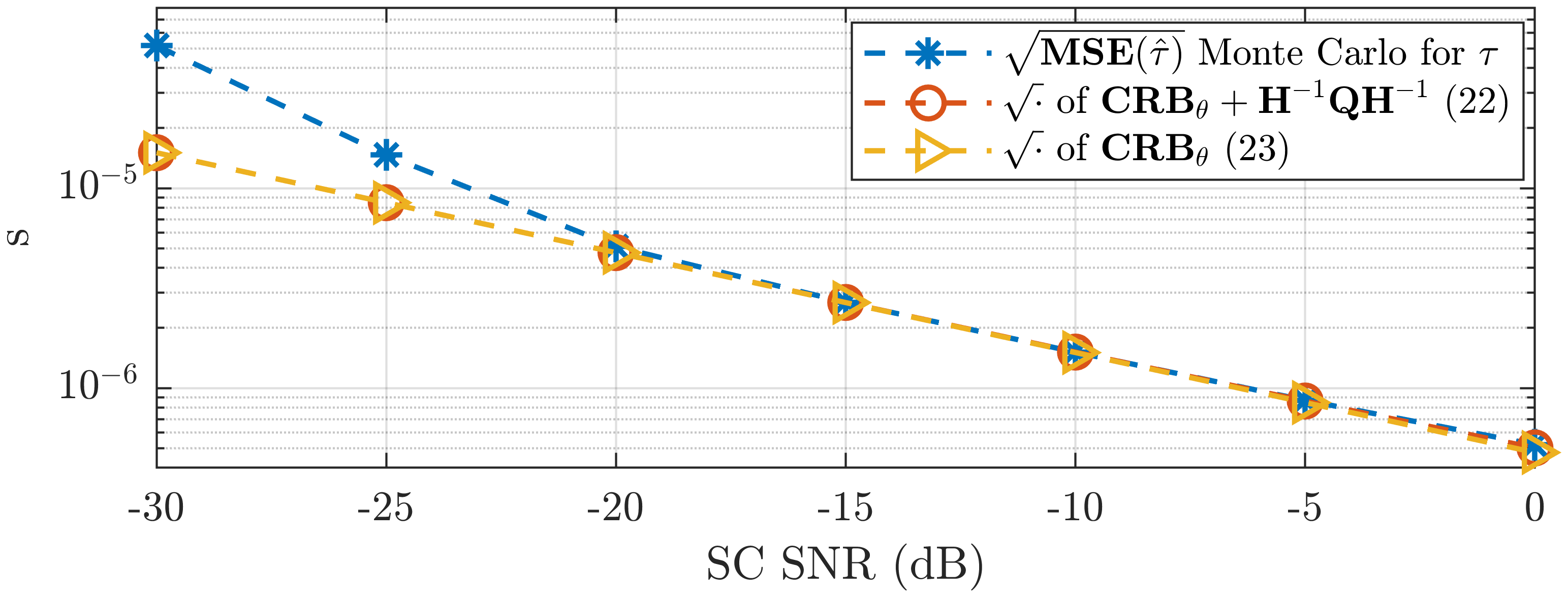}
    \vspace{-15pt}
    \caption{}
     \label{fig:RMSEtau}
  \end{subfigure}
  
  \begin{subfigure}[b]{0.47\textwidth}
    \centering
    \includegraphics[width=\linewidth]{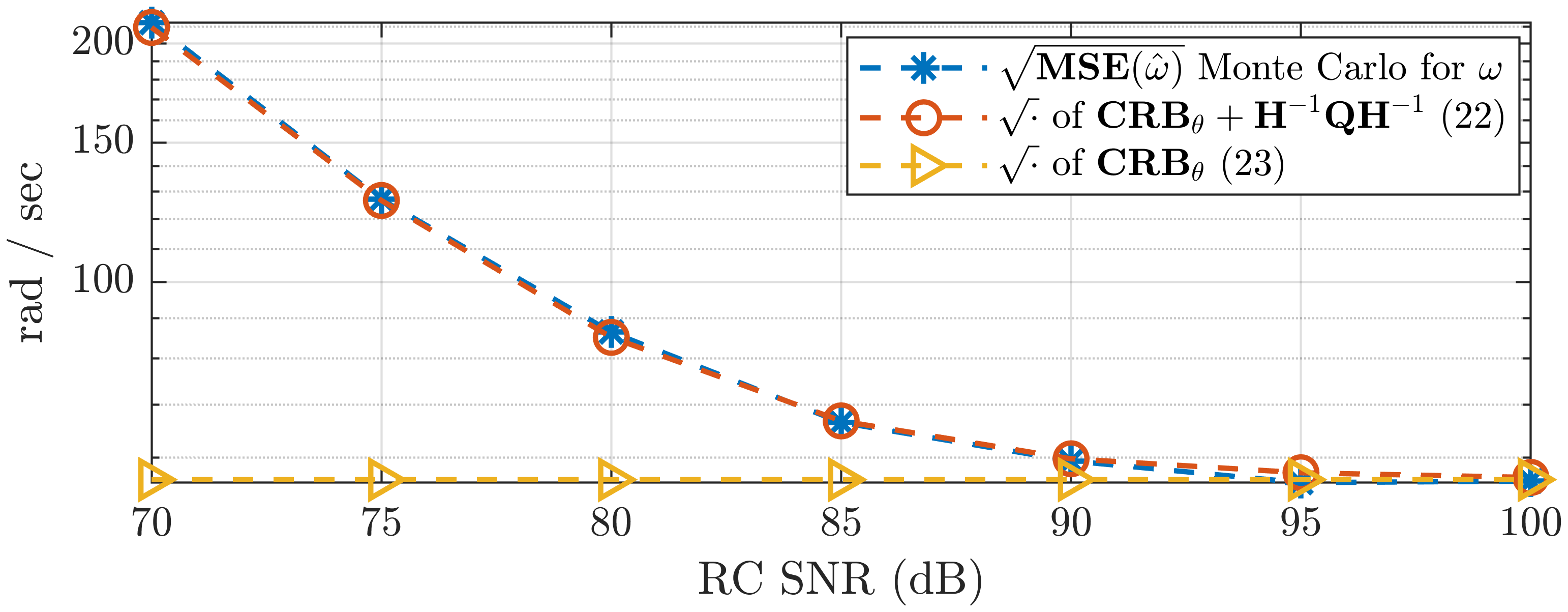}
    \vspace{-15pt}
    \caption{}
    \label{fig:RMSEomega}
  \end{subfigure}

  \caption{Monte Carlo simulation versus theoretical calculations.}
  \label{fig:both}
\end{figure}
\vspace{-10pt}
\section{Conclusions}

In this paper, we have studied the statistical performance of a method for estimating target parameters using passive radar data. An explicit expression for the covariance matrix of the estimated target 
parameters was presented. Based on this, a sufficient condition for the errors in the applied IO waveform to be negligible was given. The theoretical results were corroborated using computer simulations, showing good agreement as well as displaying the effect of IO waveform errors. A more extensive simulation study together with the results of real-data experiments will be presented in a future publication. As a final remark, we note that the theoretical results are easily extended to the frame-based processing approach of \cite{colone_multistage_2009}, which is more practical for very large data sets.

%

\bibliographystyle{ieeetr}
\bibliography{DistributedRadarRefs}

\end{document}